%% file: districtElections.tex
\def\year{2021}\relax
\documentclass[letterpaper]{article} 
\usepackage{aaai21}  
\usepackage{times}  
\usepackage{helvet} 
\usepackage{courier}  
\usepackage[hyphens]{url}  
\usepackage{graphicx} 
\usepackage{natbib}  
\usepackage{caption} 
\usepackage[switch]{lineno}
\makeatletter
    \AtBeginDocument{%
      \@ifpackageloaded{amsmath}{%
        \newcommand*\patchAmsMathEnvironmentForLineno[1]{%
          \expandafter\let\csname old#1\expandafter\endcsname\csname #1\endcsname
          \expandafter\let\csname oldend#1\expandafter\endcsname\csname end#1\endcsname
          \renewenvironment{#1}%
                           {\linenomath\csname old#1\endcsname}%
                           {\csname oldend#1\endcsname\endlinenomath}%
        }%
        \newcommand*\patchBothAmsMathEnvironmentsForLineno[1]{%
          \patchAmsMathEnvironmentForLineno{#1}%
          \patchAmsMathEnvironmentForLineno{#1*}%
        }%
        \patchBothAmsMathEnvironmentsForLineno{equation}%
        \patchBothAmsMathEnvironmentsForLineno{align}%
        \patchBothAmsMathEnvironmentsForLineno{flalign}%
        \patchBothAmsMathEnvironmentsForLineno{alignat}%
        \patchBothAmsMathEnvironmentsForLineno{gather}%
        \patchBothAmsMathEnvironmentsForLineno{multline}%
      }{}
    }
\makeatother

\usepackage{algorithm}
\usepackage[noend]{algpseudocode}
\usepackage{amssymb}
\usepackage{mathtools}
\usepackage{amsmath}
\usepackage{tikz}
\usepackage{tikz-3dplot}
\usepackage{pgfplots}
\usepgfplotslibrary{ternary}
\usetikzlibrary{calc,shapes}
\usepackage{todonotes}
\usepackage{dsfont}
\usepackage{enumitem}
\usepackage{tikz}
\usepackage{amsthm}
\usepackage{thm-restate}
\usepackage{booktabs}
\usepackage{multirow}
\usepackage{newfloat}
\usepackage{wrapfig}
\usepackage{dsfont}
\usepackage{comment}
\usepackage{multirow}
\usepackage{mleftright}

\def\cvec{\textbf{c}}
\newcommand{\C}{\mathcal{C}}
\newcommand{\W}{\mathcal{W}}

\newcommand{\ceil}[1]{\lceil #1 \rceil}
\newcommand{\set}[1]{\left\{ #1 \right\}}

\newcommand{\s}{\mathbf{s}}

\newcommand{\EX}{\mathbb{E}}

\def\eps{\epsilon}

\def\pvec{\mathbf{p}}

\def\S{\mathcal{S}}
\def\cvec{\textbf{c}}
\def\Q{\mathcal{Q}}
\def\C{\mathbf{C}}
\def\yvec{\textbf{y}}
\newcommand{\R}{\mathbb{R}}

\def\betadue{a}
\def\P{\mathcal{P}}
\def\bvec{\mathbf{b}}
\def\gammavec{\boldsymbol{\gamma}}
\def\pvar{i}
\def\hvar{l}
\def\gvar{o}
\def\qvar{v}

\newtheorem{example}{Example}

\newtheorem{lemma}{Lemma}

\newtheorem{corollary}{Corollary}
\newtheorem{proposition}{Proposition}
\newtheorem{definition}{Definition}

\title{Persuading Voters in District-based Elections}
\author {
	Matteo Castiglioni,\textsuperscript{\rm 1}
    Nicola Gatti \textsuperscript{\rm 1}
        \\
}
\affiliations {
    \textsuperscript{\rm 1} Politecnico di Milano, Piazza Leonardo da Vinci 32, I-20133, Milan, Italy \\
    matteo.castiglioni@polimi.it, nicola.gatti@polimi.it
}
\begin{document}


\maketitle

\begin{abstract}
	We focus on the scenario in which an agent can exploit his information advantage to manipulate the outcome of an election. 
	In particular, we study \emph{district-based} elections with two candidates, in which the winner of the election is the candidate that wins in the majority of the districts.
	District-based elections are adopted worldwide (\emph{e.g.}, UK and USA) and are a natural extension of widely studied voting mechanisms (\emph{e.g.}, $k$-voting and plurality voting).
	We resort to the Bayesian persuasion framework, where the manipulator (\emph{sender}) strategically discloses information to the voters (\emph{receivers}) that update their beliefs rationally. 
	We study both \emph{private} signaling, in which the sender can use a private communication channel per receiver, and \emph{public} signaling, in which the sender can use a single communication channel for all the receivers. 
	Furthermore, for the first time, we introduce \emph{semi-public} signaling in which the sender can use a single communication channel per district.
	We show that there is a sharp distinction between private and \mbox{(semi-)}public signaling.
	In particular, optimal private signaling schemes can provide an arbitrarily better probability of victory than \mbox{(semi-)}public ones and can be computed efficiently, while optimal \mbox{(semi-)}public signaling schemes cannot be approximated to within any factor in polynomial time unless $\mathsf{P}=\mathsf{NP}$. 
	However, we show that reasonable relaxations allow the design of multi-criteria PTASs for optimal \mbox{(semi-)}public signaling schemes.
	In doing so, we introduce a novel property, namely \emph{comparative stability}, and we design a bi-criteria PTAS for public signaling in general Bayesian persuasion problems beyond elections when the sender's utility function is state-dependent.
\end{abstract}

\input{introduction}
\input{preliminaries}
\input{private}

\input{stability}

\input{public}
\input{conclusions}

\section{Acknowledgments}	
	This work has been partially supported by the Italian MIUR PRIN 2017 Project ALGADIMAR ``Algorithms, Games, and Digital Market''.

\bibliography{biblio}
\clearpage
\input{appendix}
\end{document}

%% file: introduction.tex
\section{Introduction}
The fairness and efficiency of democratic elections largely depend on the  news provided by the media. 
Indeed, often, citizens are called to express opinions on complex choices they do not know deeply enough to express informed judgments. 
Therefore, multiple and reliable sources of information providing fair and in-depth coverage of the public debate are crucial to guarantee the democratic process. 
However, most of the information shaping the voters' opinions is not disclosed to inform in a disinterested way but instead aims to direct voters' political orientation, thus \emph{persuading} them to prefer one specific candidate over another.
As recently showed by \citet{allcott2017social} and \citet{guess2018selective} for the 2016 US presidential election, the spread of fake news has become a major public concern for democracy. 
The problem of assessing the extent to which it is possible to manipulate an election has received considerable attention under the general framework of election control and has been investigated according to several perspectives, such as control by bribery \cite{faliszewski2009llull, erdelyi2020complexity} or by adding and deleting candidates and voters \cite{loreggia2015controlling, faliszewski2011multimode,liu2009parameterized,chen2017elections}.
More recently, \citet{sina2015adapting}, \citet{faliszewski2018opinion}, \citet{wilder2018controlling}, \citet{wilder2019defending}, and \citet{castiglioni2020election} studied social influence as a means of election control.
In this paper, we pose the following question: \emph{can an informed agent use his information advantage to influence an election's outcome by the partial disclosure of information to rational voters? }
According to the classical Bayesian persuasion framework by \citet{kamenica2011bayesian}, the above problem can be formulated as a game with asymmetric information, where a sender can influence the behavior of the receiver(s) through the strategic provision of payoff-relevant information.
In particular, the sender can strategically reveal information by means of a signaling scheme that determines ``who knows what'' about the parameters that govern the payoff functions. 
\citet{alonso2016persuading}, \citet{chan2019pivotal}, and \citet{bardhi2018modes} provide the seminal attempts to apply the Bayesian persuasion framework to voting. 
More recently, \citet{castiglioni2019persuading} and \citet{castiglioni2020} investigated its computational issues.
All the aforementioned works focus on $k$-voting or plurality-voting elections. 
Differently, in this paper, we study for the first time how Bayesian persuasion can be adopted in more challenging settings such as district-based elections with two candidates, in which the winner of the election is the candidate winning in the majority of the districts. 
We focus on the setting with \emph{no inter-agent externalities}  where each receiver's utility depends only on his action and the realized state of nature,  but not on the other receivers' actions. 
This assumption is well-motivated, as voting for the most preferred candidate is a weakly dominant strategy in two-candidate elections. 
Two forms of signals are customarily investigated in the literature. 
With \emph{private} signals, the sender can target different information to different receivers. 
Instead, with \emph{public} signals, the sender can only communicate the same information to every receiver. 
Even if private persuasion may be more beneficial for the sender, sometimes, as in election settings where there are too many receivers, privately communicating to each receiver may be impracticable. 
At the same time, public communication is much easier to implement, \emph{e.g.}, through TVs or newspapers. 
We introduce a new form of signaling, called \emph{semi-public}, to model situations between private and public settings in district-based elections, where all the receivers of the same district observe the same signal, but the sender can target different information to different districts. 
Indeed, the voters are often reached by local communication shared with the voters in the same location, \emph{e.g.}, local newspapers, electoral posters and rallies. 

\paragraph{Original Contributions}

We study the efficiency and complexity of signaling in district-based elections with two candidates. 
First, we compare private, public, and semi-public signaling schemes in terms of efficiency when used to manipulate elections, showing that private signaling schemes perform arbitrarily better than (semi-)public schemes.
Then, we show that optimal private signaling schemes can be computed efficiently, while the direct use of the results provided by~\citet{castiglioni2019persuading} shows that the problem is inapproximable with (semi-)public signaling.
However, we prove that multi-criteria Polynomial-Time Approximation Schemes (PTASs) for public and semi-public signaling schemes are possible when some relaxations are made.
In particular, in the case of semi-public persuasion, we allow $\epsilon$-persuasiveness and lower the number of districts to control needed to win the election by an arbitrary constant factor w.r.t.~the majority.
Instead, in the case of public persuasion, we also need to lower the number of votes needed to win in a district by an arbitrary constant factor w.r.t.~the majority.
In doing so, we introduce a novel property, namely \emph{comparative stability}, and we design a bi-criteria PTAS to compute public signaling schemes in general Bayesian persuasion problems beyond district-based elections.
Our result extends that by \citet{xu2020tractability}, allowing state-dependent sender's utility functions and generalizing from stable to comparative stable sender's utility functions.

\paragraph{Related Works}

The seminal model of Bayesian persuasion with a single receiver is introduced by \citet{kamenica2011bayesian}.
This model is extended, allowing multiple receivers, by~\citet{bergemann2016bayes}, \citet{bergemann2016information}, \citet{wang2013bayesian}, and \citet{taneva2015information}.
Furthermore, \citet{alonso2016persuading}, \citet{bardhi2018modes} and \citet{chan2019pivotal} provide the first attempts of applying the Bayesian persuasion framework to voting. 
In particular, \citet{bardhi2018modes} and \citet{chan2019pivotal} study unanimity voting and $k$-voting rules, respectively, in settings with binary actions and state spaces. 
Instead, \citet{alonso2016persuading} employ a novel geometric tool to characterize an optimal public signaling scheme in voting. 
In addition to the works mentioned above, which provide the economic groundings of Bayesian persuasion in (simple) voting settings, other works study election problems from a computational perspective.
In particular, \citet{arieli2019private} study the problem of private Bayesian persuasion with no inter-agent externalities.
In the case of binary state spaces and $k$-voting rule, they provide a characterization of the optimal private signaling schemes.\footnote{In $k$-voting, a candidate wins if he collects at least $k$ votes.}
\citet{cheng2015mixture} study the same $k$-voting problem with public persuasion, providing a polynomial-time approximation algorithm for a relaxed version of the problem in which the number of votes needed to win the election is reduced by an arbitrary constant factor and $\epsilon$-persuasiveness is adopted.
\citet{castiglioni2019persuading} extend the previous models to settings with an arbitrary number of states of nature and candidates. 
In particular, they prove that a private signaling scheme for $k$-voting can be computed in polynomial time, while the optimal public signaling scheme is $\mathsf{NP}$-hard to approximate within any factor.
\citet{castiglioni2020} strengthen this hardness result, showing that finding a public signaling scheme that is approximately optimal and $\epsilon$-persuasive requires quasi-polynomial time, assuming the exponential time hypothesis. 
Our work is also closely related to Bayesian persuasion in general settings beyond elections and the relation between stability of the sender's utility function and the computation of approximations.
In particular, \citet{cheng2015mixture} introduce the notion of stability and show that this is a sufficient property to compute approximately optimal $\epsilon$-persuasive signaling schemes in polynomial time.
\citet{xu2020tractability} extends this framework to incorporate $\alpha$-approximable sender's utility functions and shows that approximately optimal and $\epsilon$-persuasive signaling schemes can be computed in polynomial time when the sender's utility function is stable and independent of the state of nature.

%% file: preliminaries.tex
\section{Problem Formulation}

In this section, we introduce the two frameworks we use in our work: district-based elections and Bayesian persuasion.

\paragraph{District-based Elections}

There is a set of candidates $C=\{c_0,c_1\}$ and a set of voters $R=\{ r_1, \dots, r_{|R|}\}$ divided in a set $D$ of districts. 
The set of voters of district $d \in D$ is denoted with $R^d$.
Each voter casts a vote for one of the two candidates.
Once the voters expressed their preferences, the election process proceeds in two steps.
For the sake of simplicity, we study the basic case in which both steps follow a \emph{majority-voting} rule.\footnote{In majority voting, the candidate with the most votes wins.}
The election works as follows.
\begin{enumerate}
	\item For each $d \in D$, the votes expressed by all $r \in R^d$ are locally aggregated, and the candidate with the majority of the votes is elected as the winner of the district.
	\item The outcomes of all the districts are aggregated, and the candidate that is the winner in the majority of the districts is chosen as the winner of the district-based election.
\end{enumerate}
We assume that the manipulator prefers $c_0$ to be the winner of the election.
Let $\cvec\in \C$ be a tuple composed by the votes of all the voters, where $\C = C^{|R|}$. 
Similarly, $\cvec^d$ is the tuple of the votes of the voters in district $d$.
The manipulator's utility $\W: \C \to \set{0,1}$ is defined as the composition of a collection of functions $W^d:C^{|R^d|}\to C$, each representing the majority voting run in district $d$, and the function $\overline{W}:C^{|D|}\to \{0,1\}$, representing the majority voting that aggregates the outcomes of all the districts. We define $K_D=\ceil{|D|/2}$ and, for each district $d$, $K_d=\ceil{|R^d|/2}$. 
Then, $\W$ is defined as  $\W(\cvec) = \overline{W}( W^1(\cvec^1), \dots , W^D(\cvec^{|D|})  )$,  where $W^d(\cvec^d)$ assumes value $c_0$ if at least $K_d$ of the voters in district $d$ vote for candidate $c_0$, and $\overline{W}$ assumes value $1$ if and only if $c_0$ wins in at least $K_D$ districts.
We introduce some relaxations for the majority-voting rules $W^d$ and $\overline{W}$.
In the first relaxation, we allow the number of votes that the target candidate $c_0$ needs to win in each district $d$ to be smaller than $K_d$.
We denote with $W^d_\delta$ the resulting majority voting rule. 
Formally, $W^d_\delta:C^{|R^d|}\to C$ assumes value $c_0$ if at least $\ceil{(1 - \delta) \, K_{d}}$ voters in district $d$ vote for $c_0$ and $c_1$ otherwise.
The manipulator's utility function of this first relaxed problem, denoted with $\W_\delta$, is defined as  $\W_\delta = \overline{W}(  W^1_\delta(\textbf{c}^1), \dots , W^D_\delta(\textbf{c}^{|D|}))$.
In the second, stronger relaxation, we also allow the number of districts that the target candidate $c_0$ needs to control to win the election to be smaller than $K_D$.
We denote with $\overline{W}_\delta$ the resulting majority voting rule aggregating the outcomes of the districts. 
Formally, $\overline{W}_\delta:C^{|D|}\to \{0,1\}$ assumes value $1$ when $c_0$ wins in at least $\ceil{(1 - \delta) \, K_{D}}$ districts.
The manipulator's utility function of this second relaxed problem, denoted with $\W_{\delta\delta}$, is defined as $\W_{\delta\delta}(\textbf{c}) = \overline{W}_\delta(  W^1_\delta(\textbf{c}^1), \ldots , W^D_\delta(\textbf{c}^D)  )$.

\paragraph{Bayesian Persuasion Framework}

Our model includes a sender (the manipulator) and a set $R$ of receivers (voters) that must choose an action (a candidate) from the set $C=\{c_0,c_1\}$.
Each voter $r$'s utility $u_r$ depends only on his own action and a state of nature $\theta \in \Theta$ drawn from a prior distribution $\mu \in \Delta_\Theta$, where $\Delta_\Theta$ is the set of probability distributions supported on $\Theta$.
In particular,  we define $u_r:\Theta\times C\rightarrow[0,1]$, where $u_r(\theta,c)$ expresses how much receiver $r$ appreciates candidate $c$ when the state of nature is $\theta$.
We use $u_r(\theta) = u_r(\theta,c_0) - u_r(\theta,c_1)$ to denote how much voter $r$ prefers candidate $c_0$ over $c_1$, in state of nature $\theta$.
In general Bayesian persuasion problems, the sender's utility, usually denoted with $f_\theta$, depends on the state of nature $\theta$ and maps the receivers' action profiles to values in $[0,1]$.
In our setting, $f_\theta$ does not depend on $\theta$ and is set equal to $\W, \W_\delta, \W_{\delta\delta}$ depending on the specific problem we tackle.
The interaction among the sender and the receivers goes as follows (see Fig.~\ref{fig:time_line}).
The sender commits to a randomized publicly known signaling scheme $\phi$ that maps states of nature to signals for the receivers.
The signal set of a receiver $r$ is denoted with $S_r$, while $s_r \in S_r$ is a signal for receiver $r$. The set of possible signals is then $\S=\times_{r\in R} S_r$, while a profile of signals is denoted with $\s=(s_1,\dots,s_{|R|})$.
The sender observes the state of nature sampled from $\mu$ and computes $\s \in \S$ according to $\phi$.
After observing the signal $s_r$, each receiver $r$ performs a Bayesian update, and infers a posterior belief $\pvec^r \in \P$ (where $\P= \Delta_\Theta$) as follows: the realized state of nature is $\theta$ with probability $p^r_\theta=\frac{\mu_\theta \,\phi(\theta,s_r)}{\sum_{\theta' \in\Theta} \mu_{\theta'}\,\phi(\theta',s_r)}$. 
Then, each receiver plays an action maximizing his expected utility according to posterior $\pvec^r$.

\begin{figure}
	\includegraphics[scale=0.74]{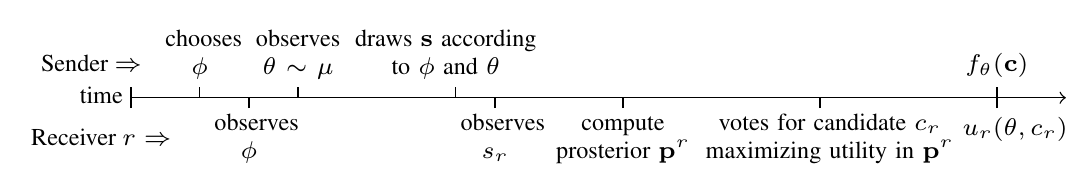}
	\caption{Interaction between the sender and a receiver.}
	\label{fig:time_line}
\end{figure}

We introduce three forms of signaling schemes. 
A private signaling scheme exploits a private communication channel toward each receiver. 
Sometimes this assumption is not realistic, and the sender has only a single communication channel observed by all the receivers, \emph{i.e.}, $s_r=s_r'$ for all $r,r' \in R$. 
We call these signaling schemes public.
Finally, we introduce a novel form of communication that suits our election model, where the sender has a communication channel toward each district $d$, and all the receivers in the same district receive the same signal, \emph{i.e.}, $s_r=s_r'$ for all $r,r' \in R^d$.
We call these signaling schemes semi-public.
In all these settings, a revelation-principle style argument shows that there always exists a signaling scheme that is \emph{direct} and \emph{persuasive}. 
More precisely, a signaling scheme is direct if the signals are action recommendations, while it is persuasive if each receiver has the interest to follow the recommendations.
Thus, a direct signaling scheme is a mapping $\phi:\Theta \rightarrow \Delta_\C$, and $\phi(\theta,\cvec)$ is the probability whereby the sender recommends $\cvec$ in state $\theta$.
In order for the signaling scheme to be persuasive, the receivers must have an incentive to follow the recommendation. 
This is customarily assured by forcing constraints on $\phi$ depending on the specific form of signaling.
In particular, the incentive constraints associated with a receiver $r$ in a district $d$ are:
\begin{itemize}
	\item $\sum_{\theta, {\cvec:c_r=c}}\phi(\theta,\cvec)(u_r(\theta,c)-u_r(\theta, c'))\ge 0 \ \forall c,c' \in C$ (private signaling);
	\item $\sum_{\theta} \phi(\theta,\cvec)(u_r(\theta,c_r)-u_r(\theta ,c'))\ge 0 \ \forall \cvec \in \C,c' \in C$ (public signaling);
	\item $\sum_{\theta, \cvec:\cvec^{d}=\bar \cvec} \phi(\theta,\cvec)(u_r(\theta,\bar{c}_r)-u_r(\theta, c'))\ge 0 \ \forall \bar{\cvec} \in C^{|R^d|}  ,c' \in C$ (semi-public signaling).
\end{itemize}
Similarly, a direct signaling scheme is $\epsilon$-persuasive if the incentive constraints are violated by at most $\epsilon$. 
Finally, we state the optimization problems we study in this paper.
\textsf{PRIVATE-DBE} is the problem of designing a private signaling scheme maximizing the probability of having candidate $c_0$ elected in district-based elections.
\textsf{PUBLIC-DBE} and \textsf{SEMIPUBLIC-DBE} refer to the same problem with public and semi-public signaling, respectively.

\paragraph{An Example of Inefficiency of (Semi-)Public Persuasion}

To clarify better the Bayesian persuasion framework, we provide an example of its application to majority voting without districts. 
This example is also useful to show that the restriction to (semi-)public signaling can decrease the sender's utility by an arbitrarily large factor.

\begin{example}\label{example}
Consider a (non-relaxed) majority-voting election with seven voters $R = \set{r_1,r_2,r_3,r_4,r_5,r_6,r_7}$ and two candidates $C = \set{c_0,c_1}$.
The objective of the sender is to maximize the probability with which candidate $c_0$ is elected. 
Therefore, he needs to persuade at least half of the voters (\emph{i.e.}, $\ceil{|R|/2}=4$) to make candidate $c_0$ be the winner.
There are three states of nature, namely, $\Theta = \set{\theta_A,\theta_B,\theta_C}$, and each state is equally probable.
Tab.~\ref{tab:limit_of_public_utilities} provides the parameters $u_r(\theta)$ of the voters, defined as $u_r(\theta) = u_r(\theta,c_0) - u_r(\theta,c_1)$ and capturing  the net payoff of voter $r$ from having candidate $c_0$ elected, in state of nature $\theta$.
 \setlength{\arrayrulewidth}{0.1mm}
 \setlength{\tabcolsep}{6pt}
 \renewcommand{\arraystretch}{1.2}
 \begin{table}[h]
 	\centering
 	\begin{tabular}{rc||c|c|c}
 		&& {State $\theta_A$}& {State $\theta_B$} & {State $\theta_C$} \\ [0.5ex]
 		\hline
 		\multirow{4}{*}{\rotatebox[origin=c]{90}{Voters}} 
 		&$r_1$,$r_2$ & $+1/2$ & $-1$ & $-1$\\
 		&$r_3$,$r_4$ & $-1$ & $+1/2$ & $-1$\\
 		&$r_5$,$r_6$ & $-1$ & $-1$ & $+1/2$\\
 		&$r_7$ & $+1/2$ & $+1/2$ & $+1/2$
 	\end{tabular}
 	\caption{Payoffs of the voters in Example \ref{example}.}
 	\label{tab:limit_of_public_utilities}
 \end{table}
The sender can design a direct and persuasive private signaling scheme such that at least four voters prefer candidate $c_0$ over $c_1$ for every signal profile $\s$. 
Hence, this scheme ensures that candidate $c_0$ is elected with a probability of 1. 
Specifically, in each state $\theta$ the scheme recommends candidate $c_0$ to every voter $r$ with utility $u_r(\theta) \ge 0$ and to one voter among those with $u_r(\theta) < 0$ chosen randomly with uniform probability.
It is easy to see that this private signaling scheme satisfies the incentive constraints.
Consider, for example, voter $r_1$.
The marginal probabilities with which he is recommended to vote for candidate $c_0$ are: $\phi_1(\theta_A,c_0) = 1, \phi_1(\theta_B,c_0) = 1/4$ and $\phi_1(\theta_C,c_0) = 1/4$.
Therefore, when he receives the recommendation to vote for $c_0$, he has a posterior distribution $\pvec$ with $p_{\theta_A}=\frac{\mu_{\theta_A} \cdot \, \phi_1(\theta_A,c_0)}{\sum\limits_{\theta \in \Theta}\mu_\theta \, \cdot \, \phi_1(\theta,c_0)}=\frac{1/3}{1/3+1/3 \cdot 1/4+1/3 \cdot 1/4}=2/3$ and $p_{\theta_B}=p_{\theta_C}=1/6$. 
Thus, the voter has expected utility  $u(\theta_A)p_{\theta_A} +u(\theta_B)p_{\theta_B}+ u(\theta_C)p_{\theta_C}=0$ and will follow the recommendation. Similarly, we can show that the incentive constraints associated with the other voters are satisfied.  
We switch to public signals and we show that we cannot design a public signaling scheme that guarantees candidate $c_0$ to be elected with positive probability.
Any public signaling scheme making candidate $c_0$ win the election with positive probability must assign a strictly positive probability to at least one signal that makes at least four voters prefer candidate $c_0$ over $c_1$.
We show that we cannot design such a public signal.
In particular, we show that there is no posterior $\pvec\in \P$ that provides an expected utility larger than or equal to zero to at least four voters.\footnote{Recall that in a public signaling scheme, all the receivers observe the same signal, perform the same update of the belief, and have the same posterior belief.}
Since receiver $r_7$ prefers candidate $c_0$ in every state of nature, he votes for $c_0$ independently from the posterior induced by the signal. 
Therefore, it is sufficient to persuade three voters among the first six.
Suppose that voters $r_1$ and $r_2$ vote for $c_0$. This implies that $p_{\theta_A}/2- p_{\theta_B}-p_{\theta_C} =p_{\theta_A}/2- (1-p_{\theta_A})\ge0$ and $p_{\theta_A}\ge2/3$.
Suppose, by contradiction, that also voters $r_3$ and $r_4$ vote for $c_0$.
This requires that $-p_{\theta_A}+ p_{\theta_B}/2-p_{\theta_C} \ge 0 $ and $p_{\theta_B}\ge 2/3$, reaching a contradiction with $\pvec \in \P$.
It is easy to see that, by the symmetry of the instance, all the other sets of four voters cannot vote for $c_0$ at the same time.
\end{example}
From the previous example, we can state the following:
\begin{proposition}\label{prop:inef}
There is an instance of majority-voting election in which the optimal private signaling scheme guarantees that candidate $c_0$ wins the election with a probability of $1$, while the optimal public signaling scheme cannot guarantee a winning probability strictly larger than $0$.
\end{proposition}

This inefficiency result can be easily generalized to the case of public and semi-public signaling scheme in district-based elections.
Indeed, with only a single district, semi-public signals correspond to public signals and a district-based election reduces to a simple majority-voting election as the one presented above.

%% file: private.tex
\section{Private Persuasion in District-based Elections} \label{sec:private}

In this section, we show that an optimal private signaling scheme for district-based elections can be found in polynomial time.
Our result is built upon the previous works by \citet{arieli2019private} and \citet{castiglioni2019persuading} on $k$-voting. 
Let $\betadue_{d,\theta}$ be the probability with which $K_d$ voters vote for $c_0$ in district $d$ when the state of nature is $\theta$. 
Similarly, let $\alpha_\theta$ be the probability that $c_0$ wins in at least $K_D$ districts with state of nature $\theta$.
Finally, given a direct private signaling scheme $\phi$, we denote with $\phi_r(\theta,c)=\sum_{\cvec:c_r=c}\phi(\theta,\cvec)$  the marginal probabilities of $\phi$ whereby $c$ is recommended to $r$ with state of nature $\theta$.
We can compute an optimal private signaling scheme by LP~\eqref{lp:private_general} (all the proofs are in the Supplemental Material).
\begin{restatable}{theorem}{theoremOne}
	\label{th: optimal_private}
	LP~\eqref{lp:private_general} computes an optimal solution of \textsf{PRIVATE-DBE} in polynomial time.
\end{restatable}
\begin{proof}[Proof sketch]
	Constraints \eqref{eq:lp_pr_persuasive_1} force the marginal probabilities $\phi_r(\theta,c_0)$ of the signaling scheme $\phi$ to satisfy the incentive constraints.
	For each state of nature $\theta$, the maximum probability $\betadue_{d,\theta}$ with which at least $K_d$ receivers in $R^d$ vote for $c_0$ given marginal probabilities $\phi_r(\theta,c_0)$ is:
	\[
	\betadue_{d,\theta}=\min\left\{\min_{m\in\{0,\dots,K_d-1\}} \frac{1}{K_d-m}\qvar_{\theta,m}; \, 1 \, \right\},
	\]
	where $\qvar_{\theta,m}$ is the sum of the lowest $|R^d|-m$ elements in the set $\{\phi_r(\theta,c)\}_{r\in R^d}$; for further details, see \citet{arieli2019private}.
	The above equation is enforced via Constraints \eqref{eq:lp_pr_beta_ub}.
	Constraints \eqref{eq:lp_pr_q_ub1} and \eqref{eq:lp_pr_phi_lb1} ensure that the values of $\qvar_{\theta,m}$ are consistent with the values of the other variables.
	The computation of the maximum probability $\alpha_\theta$ with which at least $K_D$ districts elect $c_0$ given probabilities $\betadue_{d,\theta}$ is similar to the computation of $\betadue_{d,\theta}$ given $\phi_r(\theta,c_0)$.
	This is enforced by Constraints \eqref{eq:lp_pr_alpha_ub}, \eqref{eq:lp_pr_p_ub}, and \eqref{eq:lp_pr_beta_lb}.
	Finally, Objective~\eqref{eq:lp_pr_obj} maximizes the sum over all $\theta\in\Theta$ of the prior probability multiplied by $\alpha_\theta$, \emph{i.e.}, the probability that that $c_0$ wins when the state of nature is $\theta$.
\end{proof}

\begin{subequations}\label{lp:private_general}
	\begin{align}
	&
	\max_{\substack{\alpha\in [0,1]^{|\Theta|},\, \betadue\in [0,1]^{|D|\times|\Theta|}\\
			\pvar,\hvar \in\mathbb{R}^{|\Theta|\times K_D},\, \gvar \in\mathbb{R}^{|D| \times |\Theta| \times K_D }\\
			t_{d,\theta,m},\,\qvar_{d,\theta,m}\in\mathbb{R} \ \forall d\in D,\theta\in\Theta, m\in\{1,\dots,K_d\} \\ 
			 z_{d,\theta,r,m}\in\mathbb{R} \ \forall d\in D,\theta\in\Theta, r \in R, m\in\{1,\dots,K_d\} \\
			 \phi_r(\cdot,c_0)\in[0,1]^{|\Theta|} \ \forall r \in R\\ }}  \sum_{\theta\in\Theta}\mu_\theta\,\alpha_\theta  \label{eq:lp_pr_obj}\\ 	&\textnormal{s.t.}\sum_{\theta\in\Theta}\mu_\theta\,\phi_r(\theta,c_o)\,u_r(\theta) \geq 0 \label{eq:lp_pr_persuasive_1} \hspace{2cm}\forall r\in R\\
	& \hspace{0.5cm}\alpha_\theta\leq\frac{1}{K_D - m} \pvar_{\theta,m} \label{eq:lp_pr_alpha_ub}\\ &\hspace{2.8cm}\forall\theta\in\Theta, \forall m\in\{0,\dots,K_D-1\}\nonumber\\
	& \hspace{0.5cm}\pvar_{\theta,m}\leq (|D|-m)\hvar_{\theta,m}+\sum_{d\in D}\gvar_{d,\theta,m}\label{eq:lp_pr_p_ub}\\
	& \hspace{2.8cm}\forall \theta\in\Theta,\forall m \in \{0,\dots,K_D-1\}\nonumber\\
	& \hspace{0.5cm}\betadue_{d,\theta}\geq \hvar_{\theta,m} + \gvar_{d,\theta,m} \label{eq:lp_pr_beta_lb}\\
	&\hspace{1.55cm}\forall d\in D,\forall \theta\in\Theta, \forall m\in \{0,\dots,K_D - 1\} \nonumber\\
	& \hspace{0.5cm}\betadue_{d,\theta}\leq\frac{1}{K_d-m} \qvar_{d,\theta,m}\label{eq:lp_pr_beta_ub}\\
	& \hspace{1.65cm}\forall d\in D, \forall\theta\in\Theta, \forall m\in\{0,\dots,K_d-1\}\nonumber\\
	& \hspace{0.5cm}\qvar_{d,\theta,m}\leq (|R^d| -m)t_{d,\theta,m}+\sum_{r\in R^d}z_{d,\theta,r,m}\label{eq:lp_pr_q_ub1}\\
	& \hspace{1.65cm}\forall d\in D,\forall \theta\in\Theta,\forall m \in \{0,\dots,K_d-1\}\nonumber\\
	& \hspace{0.5cm}\phi_r(\theta,c_0)\geq t_{d,\theta,m} + z_{d,\theta,r,m} \label{eq:lp_pr_phi_lb1}\\
	& \hspace{1cm}\forall d \in D, \forall r\in R^d,\forall \theta\in\Theta, \forall m\in \{0,\dots,K_d-1\} \nonumber	
	\end{align}
\end{subequations}

%% file: stability.tex
\section{Public and Semi-public Persuasion in District-based Elections}

We turn our attention to the design of optimal public and semi-public signaling schemes. 
There is a sharp distinction between the nature of these problems and that one of private signaling. 
Indeed, in addition to being inefficient w.r.t.~private signals (see Proposition \ref{prop:inef}), optimal \mbox{(semi-)}public signaling schemes are also inapproximable. 
The hardness follows from previous results with public signaling.
Specifically, \citet{castiglioni2019persuading} prove that it is $\mathsf{NP}$-hard to approximate the optimal public signaling scheme within any factor in elections with majority voting.
The extension of this hardness result to public and semi-public signaling in district-base elections is direct as a district-based election reduces to majority-voting when there is only a single district.
Thus, we focus on possible relaxations that make the problem computationally tractable.
Motivated by the fact that voters are somewhat biased to follow the sender's recommendations, several works relax the incentive constraints allowing the receivers to vote for the target candidate even if other candidates give them a slightly better expected utility ($\eps$-persuasiveness).
Recently, \citet{castiglioni2020} prove that even allowing this relaxation the problem of designing an approximate public signaling scheme remains intractable with majority voting. 
Therefore, we focus on other different relaxations.
In particular, \citet{cheng2015mixture} employ two forms of relaxation, adopting $\epsilon$-persuasiveness and lowering the number of votes needed to win the election by an arbitrary constant factor. 
With these two relaxations, they prove that an approximate public signaling scheme with majority-voting can be computed efficiently.
We prove that, adapting these two relaxations to our settings, both \textsf{PUBLIC-DBE} and \textsf{SEMIPUBLIC-DBE} admit a multi-criteria PTAS.
As a preliminary step, we prove some results on the relation between the notion of stability and the design of approximately optimal signaling schemes that are of general interest in Bayesian persuasion beyond elections.

\paragraph{Comparative Stability and Public Signaling Schemes}

We refer to the notion of stability of a function introduced by~\citet{xu2020tractability}.
In particular, a function is said stable if, for every action profile, the introduction of small perturbations leads to small changes in the value of the function.
Here, we extend the notion of stability to pairs of functions, and we call it comparative. 
Our extension is such that comparative stability corresponds to (simple) stability in the degenerate case in which the two functions of the pair are the same. 
Furthermore, if function $g$ satisfies the comparative stability property w.r.t.~function $h$, we also say that $g$ is $\beta$-stable compared with $h$.
Initially, we introduce the notion of perturbation by the concept of $\alpha$-noisy distribution.
\begin{definition} 
	Let $\cvec \in \C$ be an action profile and $\yvec$ be a probability distribution supported on $\Delta_\C$.
	For any $\alpha \in (0,1]$, we say that $\yvec$ is an $\alpha$-noisy distribution around $\cvec$ if for all $i \in \set{1, \ldots, n} : \Pr_{\tilde \yvec \sim \yvec}[\tilde y_i \neq c_i] \leq \alpha$.
\end{definition}
Hence, an $\alpha$-noisy distribution bounds the \emph{marginal probability} of any single element of $\set{1,\ldots, n}$ to be corrupted.
However, no assumption is made on how the corruptions of the elements correlate with each other.
Now, we define our notion of comparative stability. 
\begin{definition}
\label{def: Comparative Stability}
Given two functions $g,h : \C \to [0,1]$ and a real number $\beta \geq 0$, we say that $g$ is $\beta$-stable compared with $h$ if and only if the following holds for all action profiles $\cvec \in \C$, $\alpha \in (0,1]$, and $\alpha$-noisy distributions $y$ around $\cvec$:
\[
\underset{\tilde{\yvec} \sim \yvec}\EX\left[ g(\boldsymbol{\tilde y})\right]\geq h(\cvec) (1 - \alpha \beta).
\]
\end{definition}
Intuitively, if $g$ satisfies the comparative stability property w.r.t.~$h$, then, for every action profile, the value of $h$ in that action profile is close to the value of $g$ in the corresponding perturbed action profile.
We exploit the notion of comparative stability to design an efficient algorithm that computes approximate public signaling schemes.
More precisely, we study a generic multi-agent Bayesian persuasion problem, where the sender faces a set of receivers $R$, and each receiver needs to choose an action between a couple of alternatives.
Let $g, h$ be two sets of arbitrary functions depending on the state of nature $\theta$ and denoted with $g_\theta: \C \rightarrow [0,1]$ and $h_\theta: \C \rightarrow [0,1]$, respectively.
According to Definition~\ref{def: Comparative Stability}, we say that $g$ is $\beta$-stable compared with $h$ if $g_\theta$ is $\beta$-stable with respect to $h_\theta$ for all the states of nature $ \theta \in \Theta$.
For the sake of clarity, in the following, we use indirect signaling schemes, and we express a signaling scheme as a weighted set of posteriors to which the receivers respond at best. 
Now, we describe the optimal behavior of the receivers.
\begin{definition}[Receivers' behavior with persuasiveness]\label{def:br_set}
	Given a set of functions $\{f_\theta\}_{\theta\in\Theta}$ such that $f_\theta : \C \to [0,1]$, the receivers' optimal behavior $\bvec^\pvec  \in \C$ with persuasiveness given posterior $\pvec\in \P$ is  as follows. 
	Let:
	\begin{itemize}
		\item A = $\left\{ r \in R : \sum_\theta p_\theta \, u_r(\theta)>0\right\}$ the set of receivers whose  unique best response is action $c_0$,
		\item B = $\left\{ r \in R : \sum_\theta p_\theta \, u_r(\theta)<0\right\}$ the set of receivers whose  unique best response is action $c_1$,
		\item E = $\left\{ r \in R : \sum_\theta p_\theta \, u_r(\theta)=0\right\}$ the set of receivers who are indifferent between action $c_0$ and $c_1$.
	\end{itemize}
	Then, we have: 
	\[
	\bvec^\pvec \, = \arg\max_{\cvec \in \C : c_r = c_0\forall r \in A, \, c_r = c_1  \forall r \in B}\sum_{\theta} p_\theta \, f_\theta(\cvec).
	\]
\end{definition}

Similarly, we define the notion of $\epsilon$-best response. 
\begin{definition}[Receivers' behavior with $\eps$-persuasiveness]\label{def:eps_br_set}
	Given a set of functions $\{f_\theta\}_{\theta\in\Theta}$ such that \mbox{$f_\theta : \C \to [0,1]$,} the receivers' optimal behavior $\bvec^{\pvec,\eps}  \in \C$ with $\eps$-persuasiveness  given posterior $\pvec\in \P$ is  as follows. 
	Let:
	\begin{itemize}
		\item $A_\eps = \left\{ r \in R : \sum_\theta p_\theta \, u_r(\theta)>\eps\right\}$ the set of receivers whose  unique best response is action $c_0$,
		\item $B_\eps = \left\{ r \in R : \sum_\theta p_\theta \, u_r(\theta)<- \eps\right\}$ the set of receivers whose  unique best response is action $c_1$,
		\item $E_\eps = \left\{ r  \in R : \sum_\theta p_\theta \, u_r(\theta)\in[-\eps,\eps]\right\}$ the set of receivers who are indifferent between action $c_0$ and $c_1$.
	\end{itemize}
	Then, we have:
	 \[
	 \bvec^{\pvec,\epsilon}  = \arg\max_{\cvec \in \C: c_r = c_0\forall r \in A_\epsilon \, c_r = c_1  \forall r \in B_\epsilon}\sum_{\theta}p_\theta \, f_\theta(\cvec).
	 \]
\end{definition}

Now, we show that computing a direct public signaling scheme is equivalent to derive a Bayes plausible distribution of posteriors $\gammavec \in \Delta_\P$ that maximizes the sender's utility.
Let $supp(\gammavec)$ denote the set of posteriors induced with strictly positive probability.
Similarly, let $supp(\phi)$ denote the set of posteriors induced by $\phi$ with strictly positive probability.
Finding a public signaling scheme is equivalent to finding a probability distribution $\gammavec \in \Delta_\P$ on the set of posteriors $\P$ such that $\sum_{p \in supp(\gammavec)} \gamma_p \, p_\theta=\mu_\theta$ for every $\theta\in \Theta$.
Given a well-defined distribution over posteriors $\gamma$, we can recover a direct signaling schemes $\phi$ that induces such a probability distribution by setting $\phi_\theta(\textbf{c})=\sum_{\pvec \in supp(\gammavec): \cvec=\bvec^\pvec} \gamma_p \,p_\theta$.
For this reason, in the following, we represent signaling schemes as probability distributions on the posteriors.
We introduce some further notation.
For every $\pvec \in \P$ and set of functions $f=\{f_\theta\}_{\theta \in \Theta}$, we define the sender's expected utility with persuasiveness as $f(\pvec) = \sum_{\theta}\pvec_\theta  \,f_\theta(\bvec^\pvec)$, and with $\eps$-persuasiveness as $f_\eps(\pvec) = \sum_{\theta}\pvec_\theta  \,f_\theta(\bvec^{\pvec,\eps})$.
Finally, we define $q$-uniform probability distributions as follows.
\begin{definition}\label{def:s_distribution}
A probability distribution $\mathbf{x} \in\Delta_X$ is \em $q$-uniform if and only if it is the average of a multiset of $q$ basis vectors in $|X|$-dimensional space.
\end{definition}
Therefore, we say that a probability distribution $\pvec \in \P$ is $q$-uniform if each of its entry $p_\theta$ is a multiple of $1/q$. 
Moreover, we use the notation $\Q \subset \Delta_\Theta$ to denote the set of all $q$-uniform distributions over $\Theta$.
Our first result shows that we can decompose each posterior in a convex combination $\gammavec \in \Delta_\Q$ of $q$-uniform posteriors (with $q$ constant), such that 
 $\sum_{\pvec \in \Q} \gamma_p\,  g_\eps(\pvec)$ closely approximates $h(\pvec^\ast)$.
This is a generalization of the result by \citet{xu2020tractability} to state-dependent utility functions (and couples of functions), and it is crucial to prove the following results.
\begin{restatable}{lemma}{lemmaOne}\label{th: general bound}
Let $\beta, \eps > 0,  \eta \in (0,1]$ and set $q= 32\log\left(\frac{4}{\eta\min\{1;\,1/\beta\}}\right)/\eps^2$.
Then, given a posterior $\pvec^\ast \in \P$ and two sets of functions $g,h$ with $g$ $\beta$-stable compared with $h$, there exists a $\gammavec \in \Delta_\Q$ with $\sum_{p \in \Q} \gamma_\pvec \,\pvec=\pvec^*$ and
\begin{equation}\label{eq:general_bound}
	\sum_{p \in \Q} \gamma_\pvec \sum_\theta p_\theta \, g_\theta(\bvec^{\pvec,\eps})\ge (1-\eta)\sum_\theta p^\ast_\theta \, h_\theta(\bvec^{\pvec^\ast}).
\end{equation}
\end{restatable}

Now, we can prove the main result of this section.
Consider a couple of sets of functions $g,h$ where $g$ is $\beta$-stable compared with $h$. 
With abuse of notation, we define $g(\phi)$ and $h(\phi)$ as the functions which evaluate the expected sender's utility of a public signaling scheme $\phi$ with $h$ and $g$, respectively.
We can resort to Lemma~\ref{th: general bound} to state the following result.
The proof is based on solving a linear program that works only with $q$-uniform posteriors.
\begin{restatable}{theorem}{theoremTwo}\label{th: bi-criteria general}
Let $\beta,\eps > 0$ and $\eta \in (0,1]$. Consider two arbitrary state-dependent sets of functions $g, h$ such that $g_\theta: \C \to [0,1]$ is $\beta$-stable compared with $h_\theta:\C \to[0,1]$ for all $\theta \in \Theta$.
Then there exists a $poly\left(|R| \; |\Theta|^{\log( \frac{1}{\eta \, \min\{1; 1/\beta\}} )/\eps^2} \right)$ time algorithm that returns an $\eps$-persuasive public signaling scheme $\phi_\eps$ such that:
\begin{equation}
g(\phi_\eps) \ge (1 - \eta) \max_{\phi \in \Phi} h(\phi), \nonumber    
\end{equation}
where $\Phi$ is the set of persuasive signaling schemes.
\end{restatable}
By setting $h=g$, we obtain a generalization of the result by~\citet{xu2020tractability} to state-dependent functions.

\paragraph{Comparative Stability of Voting Functions}

We apply this novel concept of stability to voting problems.
Our first result proves that the two relaxed majority-voting functions previously introduced satisfy the comparative stability property.
This result is similar to that by \citet{cheng2015mixture}. 
However, we use multiplicative factors (in place of additive factors) and prove a slightly stronger result than stability. In particular, we prove that the decrease in utility is small even if only the perturbations from action $c_0$ to $c_1$ are bounded.
\begin{restatable}{lemma}{lemmaTwo}
\label{lemma: stability of W_delta wrt W}
$W_\delta$ is $1/\delta$-stable compared with $W$. Moreover, for all $\textbf{c} \in \C$, $r \in R$, $\alpha \in (0,1]$, and $\yvec \in \Delta_\C$ such that $\Pr_{\yvec}\left(\,\tilde y_r = c_1 \land c_r = c_0\,\right) \leq \alpha$, it holds: 
\begin{equation}
\EX_{\tilde \yvec \sim{} \yvec}\left[ W_\delta(\boldsymbol{\tilde{y}})\right] \geq W(\boldsymbol{c}) \, \left(1-\frac{\alpha}{\delta}\right). 	\nonumber
\end{equation}
\end{restatable}

We can use the result above to prove that $\W_{\delta\delta}$ satisfies the property of comparative stability with respect to $\W$. 
Intuitively, the result follows from the observation that $\W$ is the composition of two majority-voting steps.
\begin{restatable}{lemma}{lemmaThree}
\label{th:stability W_district and W_district_delta}
$\mathcal{W}_{\delta\delta}$ is $\frac{1}{\delta^2}$-stable with respect to $\mathcal{W}$.
\end{restatable}

Finally, we derive a stronger decomposition lemma for majority-voting. 
Specifically, Lemma~\ref{th: general bound} shows that the decrease in the expected sender's utility when decomposing a posterior in $q$-uniform posteriors can be bounded. 
%
%
However, in generic settings, the sender's expected utility in a given state of nature can change arbitrarily. This is not the case in majority voting, where, instead, this decrease is bounded.
In particular, we can show the following, that is crucial when addressing the \textsf{SEMIPUBLIC-DBE} problem.
\begin{restatable}{lemma}{lemmaFour}\label{corollary: th_general bound}
Let $\eps > 0,  \eta \in (0,1]$ and set $q= 32\log\left(\frac{4}{\eta \delta} \right)/\eps^2$.
Then, given a posterior $\pvec^\ast \in \P$, there exists a $\gammavec \in \Delta_\Q$ with $\sum_{\pvec \in \Q} \gamma_\pvec \, \pvec=\pvec^*$ and
\[
	\sum_{\pvec \in \Q} \gamma_\pvec  \,p_\theta \,W_\delta(\bvec^{\pvec,\eps}) \ge (1-\eta)\, p^\ast_\theta  \,W(\bvec^{\pvec^\ast}) \ \forall \theta \in \Theta.
\]
\end{restatable}

%% file: public.tex
\paragraph{Computing Public and Semi-public Signaling Schemes in District-based Elections}
We present two multi-criteria PTASs for the \textsf{SEMIPUBLIC-DBE} and \textsf{PUBLIC-DBE} problems, respectively, when our relaxations are adopted.
First, we focus on the problem of designing public signaling schemes.
We assume $\eps$-persuasive signaling schemes, and we replace function $\W$ with $\W_{\delta\delta}$ (this corresponds to relaxing both the majority voting within every single district and the majority voting aggregating the outcomes of all the districts).
%
%
Let $\W(\phi)$ and $\W_{\delta\delta}(\phi)$ denote the functions returning the sender's expected utility provided by a public signaling scheme $\phi$ with voting rules $\W$ and $\W_{\delta\delta}$, respectively.
We show that it is possible to compute efficiently an $\eps$-persuasive public signaling scheme $\phi_\eps$ that approximates the optimal persuasive signaling scheme with an approximation factor arbitrarily close to 1.
Since the relaxed function $\W_{\delta\delta}$ is $1/\delta^2$-stable compared to the non-relaxed function $\W$ by Theorem \ref{th:stability W_district and W_district_delta}, we can immediately apply Theorem \ref{th: bi-criteria general} to these functions and then derive the following.
\begin{corollary}
\label{th: bicriteria PTAS for PUBLIC-DBE}
Let $\eps > 0$, $\delta \in (0,1)$ and $\eta \in (0,1]$, then there exists a $poly\left(|R| \; |\Theta|^{\log\left( \frac{1}{\eta \, \delta^2} \right)/\eps^2} \right)$ time algorithm that returns an $\eps$-persuasive public signaling scheme $\phi_\eps$ such that:
\begin{equation}
\W_{\delta\delta}(\phi_\eps) \ge (1 - \eta) \, \max_{\phi \in \Phi}	\W(\phi),
\end{equation}
where $\Phi$ is the set of persuasive signaling schemes.
\end{corollary}
Then, we focus on the \textsf{SEMIPUBLIC-DBE} problem.
As highlighted above, to overcome the intractability result, also in this setting, it is necessary to relax the problem.
Specifically, we use $\eps$-persuasive signaling schemes and we replace function $\W$ with $\W_{\delta}$ (this corresponds to relaxing the majority voting aggregating the outcomes of all the districts).
We show that it is possible to compute efficiently an $\eps$-persuasive semi-public signaling scheme $\phi_\eps$ that approximates the optimal persuasive signaling scheme with an approximation factor arbitrarily close to 1.
Computing a semi-public signaling scheme $\phi$ amounts to determining a collection $\{\phi_d\}_{d\in D}$ of $|D|$ public signaling schemes, one for each district, and correlate them.
The crucial point concerns the computation of good marginal probabilities of the signaling scheme. 
Indeed, their aggregation is equivalent to computing a private signaling scheme in majority-voting elections, and this can be done efficiently (see LP \eqref{lp:private_general} and Theorem \ref{th: optimal_private}).
The main idea of our proof is that there are approximately optimal marginal probabilities of the signaling scheme that use only $q$-uniform posteriors (with $q$ constant). 
Let $\alpha_\theta$ be the probability that $c_0$ wins in at least $K_D$ districts with state of nature $\theta$, $\betadue^\delta_{d,\theta}$ be the probability that candidate $c_0$ receives at least $\ceil{(1-\delta)\, K_d}$ votes in district $d$ with state of nature $\theta$, and $\gammavec^d$ be a probability distribution over posteriors for the receivers in district $d$. 
Finally, let $\mathbb{I}[\mathcal{E}]$ denote the indicator function for the event $\mathcal{E}$.
Then, the following formulation computes an approximately optimal signaling scheme in polynomial time.

\begin{subequations}\label{lp:semi_efficient}
	\begin{align}
	&\max_{\substack{\alpha\in [0,1]^{|\Theta|},\, \betadue^\delta \in [0,1]^{|D|\times|\Theta|}\\
                    \pvar,\hvar\in\mathbb{R}^{|\Theta|\times K_D},\, \gvar\in\mathbb{R}^{|D| \times |\Theta| \times K_D}\\
                    \gammavec^d \in \Delta_{\Q} \forall d \in D}} \sum_{\theta\in\Theta}\mu_\theta\alpha_\theta \label{eq:semi_obj}\\
	&\textnormal{s.t. }\alpha_\theta\leq\frac{1}{K_D - m} \pvar_{\theta,m}\label{eq:semi_alpha_ub}\\ &\hspace{2.4cm}\forall\theta\in\Theta, \forall m\in\{0,\ldots,K_D-1\} \nonumber\\
	&\hspace{0.6cm} \pvar_{\theta,m}\leq (|D|-m)\hvar_{\theta,m}+\sum_{d\in D}\gvar_{d,\theta,m}\label{eq:semi_p_ub} \\
	&\hspace{2.4cm}\forall \theta\in\Theta,\forall m \in \{0,\ldots,K_D-1\}\nonumber\\
	&\hspace{0.6cm}\betadue^\delta_{d,\theta}\geq \hvar_{\theta,m} + \gvar_{d,\theta,m} \label{eq:semi_beta_lb}\\
	&\hspace{1.2cm}\forall d\in D,\forall \theta\in\Theta, \forall m\in \{0,\ldots,K_D - 1\} \nonumber\\
	&\hspace{0.6cm} \betadue^\delta_{d,\theta} \le \sum_{\pvec \in \Q} \frac{\gamma_\pvec^d \, p_\theta}{\mu_\theta}  \mathbb{I}\left(W_\delta^d(\bvec^{\pvec,\eps})=c_0\right) \label{eq:semi_beta_ub}\\
	&\hspace{4.7cm}\forall d\in D,\forall \theta\in\Theta \nonumber\\
	& \hspace{0.6cm}\sum_{\pvec \in \Q} \gamma_\pvec^d \, p_\theta = \mu_\theta \label{eq:semi_p_bayes_rationality}
	\hspace{1.8cm} \forall d\in D,\forall \theta\in\Theta
	\end{align}
\end{subequations}

\begin{restatable}{theorem}{theoremThree}
\label{th: bicriteria PTAS for SEMIPUBLIC-DBE}
Let $\eps > 0$, $\delta \in (0,1)$ and $\eta \in (0,1]$, then there	exists a $poly\left( |R| \; |\Theta|^{\log\left( \frac{1}{\eta \, \delta} \right)/\eps^2} \right)$ time algorithm that outputs an $\eps$-persuasive semi-public signaling scheme $\phi_\eps$ such that:
\begin{equation}
\W_\delta(\phi_\eps) \ge (1 - \eta) \, \max_{\phi \in \Phi} \W(\phi),
\end{equation}
where $\Phi$ is the set of persuasive signaling schemes.
\end{restatable}

%% file: conclusions.tex
\section{Conclusions and Future Works}

In this paper, we study how a manipulator can exploit his information advantage to manipulate a district-based election through the strategic provision of information to rational voters.
We show that private signaling schemes can be computed efficiently while computing optimal \mbox{(semi-)}public signaling schemes is intractable.
However, we show that reasonable relaxations allow the design of multi-criteria PTASs for \mbox{(semi-)}public persuasion.
An interpretation of these relaxations is that \mbox{(semi-)}public signaling is often tractable, except when the target candidate wins in at least half of the districts, but it is impossible to make slightly more than half of them elect such a candidate.
In most cases, the receivers are slightly biased to follow the sender recommendations, and the manipulator's preferred candidate can either win or not win by at least a small, but not negligible, margin. 
With these assumptions, our algorithm approximates arbitrarily well the optimal signaling scheme in polynomial time.

In the future, we will study classes of instances in which optimal \mbox{(semi-)}public signaling schemes can be computed efficiently.
We are also interested in settings in which the sender is uncertain about the voters' preferences.

%% file: appendix.tex
\section{Omitted Proofs on ``Private Persuasion''}
\theoremOne*

\begin{proof}

	LP~\eqref{lp:private_general} has a polynomial number of variables and constraints and, therefore, it can be solved in polynomial time.
	Thus, we just need to prove that LP~\eqref{lp:private_general} actually computes an optimal solution to \textsf{PRIVATE-DBE}.
	First, we remark that all the marginal probabilities $\phi_r(\theta,c_0)$ of the signaling scheme $\phi$ must satisfy the incentive Constraints~\eqref{eq:lp_pr_persuasive_1}.
	$\betadue_{d,\theta}$ represents the probability of having at least $K_d$ votes in district $d$, given state of nature $\theta$.
	We need to show $\betadue_{d,\theta}$ is computed correctly given the other variables of LP~\eqref{lp:private_general}. 
	In particular, for every state of nature $\theta$, the maximum probability with which at least $K_d$ of the receivers in $R^d$ vote for $c_0$ given marginals probabilities $\phi_r(\theta,c_0)$ is:
	\[
	\betadue_{d,\theta}=\min\left\{\min_{m\in\{0,\dots,K_d-1\}} \frac{1}{K_d-m}\qvar_{\theta,m}; \, 1 \, \right\},
	\]
	where $\qvar_{\theta,m}$ is the sum of the lowest $|R^d|-m$ elements in the set $\{\phi_r(\theta,c)\}_{r\in R^d}$;  further details are provided by~\cite{arieli2019private}.
	This definition is encoded by Constraints~\eqref{eq:lp_pr_beta_ub}.
	Constraints~\eqref{eq:lp_pr_q_ub1} and~\eqref{eq:lp_pr_phi_lb1} ensure the values $\qvar_{\theta,m}$ are well defined and derived from the dual of a simple LP of this
	kind: 
	\begin{align*}
		\min_{\mathbf{y}\in\R^n} 		& \textbf{x}^\top \mathbf{y}				\\
								& \mathbf{1}^\top \mathbf{y}=w								\\
								& \mathbf{0} \leq \mathbf{y} \leq \mathbf{1}
	\end{align*}
	where $\mathbf{x}\in\R^n$ is the vector from which we want to extract the sum of the smallest $w$ entries.
	Finally, we prove that $\alpha_\theta$  is computed correctly. 
	The computation of the maximum probability $\alpha_\theta$ with which at least $K_D$ districts elect $c_0$ given probabilities $\betadue_{d,\theta}$ is similar to the computation of $\betadue_{d,\theta}$ given $\phi_r(\theta,c_0)$.
	For a similar argument as above, Constraints \eqref{eq:lp_pr_alpha_ub}, \eqref{eq:lp_pr_p_ub}, and \eqref{eq:lp_pr_beta_lb}  correctly compute $\alpha_\theta$ aggregating the marginal probabilities $\{\betadue_{d,\theta}\}_{d\in D, \theta \in \Theta}$.
	Objective~\eqref{eq:lp_pr_obj} is given by the sum over all $\theta\in\Theta$ of the prior of state $\theta$, multiplied by $\alpha_\theta$. 
	Thus, by definition of $\alpha_\theta$, we are maximizing the probability of having $c_0$ locally elected in more than $K_D$ districts.
	Finally, we prove how to construct a signaling scheme $\phi'$ with the same objective function of LP~\eqref{lp:private_general}. 
In particular, we find marginal signaling schemes $\phi_r'$ such that the incentive constraints relative to $c_0$ and $c_1$ are satisfied and $\phi_r'(\theta,c_0)\ge \phi_r(\theta,c_0)$ for all $r$ and $\theta$.
Since we do not introduce the incentive constraint relative to action $c_1$, they could not be satisfied by $\phi$.
	However, from the optimal marginal probabilities $\phi_r(\theta,c_0)$, it is straightforward to compute the marginal probabilities $\{\, \phi'_r(\theta, c_0),\, \phi'_r(\theta, c_1) \,\}_{r\in R, \theta \in \Theta}$.
	For each state of nature $\theta$, let $\phi'_r(\theta, c_0)=1$ if $u_r(\theta) \geq 0$ and $\phi'_r(\theta,c_0)=\phi_r(\theta)$ otherwise. 
	Then, $\phi'_r(\theta,c_1)=1-\phi'_r(\theta,c_0)$.
	The marginal signaling scheme $\phi_r'$ is persuasive as $c_1$ is recommended only when it is the optimal action, while $\phi_r'(\theta, c_0)\ge \phi_r(\theta, c_0)$ if and only if $u_\theta\ge 0$. 
	Formally, $\sum_{\theta\in\Theta}\mu_\theta\,\phi'_r(\theta,c_0)\,u_r(\theta) \geq \sum_{\theta\in\Theta}\mu_\theta\,\phi_r(\theta,c_0)\,u_r(\theta) \ge 0$ by constraints \eqref{eq:lp_pr_persuasive_1}.
	Finally, we can aggregate the marginal probabilities of the signaling scheme by using the same approach proposed by~\citet{arieli2019private}.
	%
\end{proof}

\section{Omitted Proofs on ``Comparative Stability and Public Signaling Schemes''}

\lemmaOne*

\begin{proof} 
	Let $\tilde\gammavec \in \Q$ be the empirical distribution of $q$ i.i.d. samples drawn from $\pvec^\ast$, where each $\theta$ has probability $p^\ast_\theta$ of being sampled. 
	Therefore, the vector $\tilde\gammavec$ is a random variable supported on $q$-uniform posteriors with expectation $\pvec^\ast$. 
	Moreover, let $\gammavec\in\Delta_{\Q}$ be a probability distribution such as, for every $p\in\Q$, it holds $\gamma_\pvec = \Pr(\tilde\gammavec=\pvec)$.
	It is easy to see that $\pvec^\ast=\sum_{\pvec\in \Q}\gamma_{\pvec} \pvec$.
	We need to prove that Equation~(\ref{eq:general_bound}) holds.
	For every $\pvec \in\Q$, we define with $\gamma_\pvec^{(\theta,i)}$ the conditional probability of having observed posterior $\pvec$ given that the posterior assigns a probability of $i/q$ to state $\theta$.
	Formally, for every $\pvec \in\Q$, we have:
	\begin{equation}
	\gamma_\pvec^{(\theta,i)} =
	\begin{cases}
	\displaystyle\frac{ \gamma_\pvec}{\sum\limits_{\pvec'\in \Q:p'_\theta=i/q}\gamma_{\pvec'}} & \text{if $p_\theta=i/q$}\\
	& \\
	\hspace{1cm}0 & \text{otherwise}
	\end{cases}.
	\nonumber
	\end{equation}
	Then, the random variable $\tilde\gammavec^{(\theta,i)}\in\Q$ is such that, for every $\pvec \in\Q$, it holds $\Pr(\tilde\gammavec^{(\theta,i)}=\pvec)=\gamma^{(\theta,i)}_\pvec$.
	For each $r\in R$, we define $\mathcal{P}^r \subseteq \Q$ as the set of posteriors that do not change the expected utility of $r$ by more than $\eps$ with respect to $\pvec^\ast$. 
	Formally, $\pvec \in \mathcal{P}^r$ if and only if $|\sum_\theta p_\theta \, u_r(\theta) -\sum_\theta \, p^\ast_\theta u_r(\theta) |\le \eps$.
	Finally, let $\alpha = \eta\,\min\{1;\,1/\beta\}$.

	To complete the proof, we introduce the following three lemmas.
	First, given a probability distribution $\pvec^\ast$ and a state of nature $\theta \in \Theta$, the following lemma bounds the maximum probability mass that $\gammavec$ assigns to posteriors $\pvec \in \Q$ in which the probability assigned to state of nature $\theta$ deviates from the one prescribed by $\pvec^\ast$ by at least $\eps/4$. 
	\begin{lemma}\label{lemma: Th.3 lemma_1}
		Given $\pvec^\ast\in \P$, for each $\theta \in \Theta$, it holds:
		\[
		\sum_{i : |i/q-p^\ast_\theta|\ge \eps/4} \; \sum_{\pvec \in \Q:p_\theta=i/q} \gamma_p \le \frac{\alpha}{2} \, p^\ast_\theta,
		\]
		where $\gammavec$ is the probability distribution of $q$ i.i.d samples drawn from $\pvec^\ast$.
	\end{lemma}
	\begin{proof}

		We observe that the random variable $\tilde\gamma_{\theta}$ is drawn from a Binomial probability distribution.
		We consider two possible cases.
		If $p^\ast_\theta \ge 1/8$, then by Hoeffding's inequality we can write the following:
		\begin{subequations} 
			\begin{align}
			\Pr\left(|\tilde{\gamma}_\theta-p^\ast_\theta|\ge\frac{\eps}{4}\right)  &\le 2 \,e^{-2\,q\, (\epsilon/4)^2}=\\
			&= 2\,e^{-4\log(4/\alpha )}\le \\
			&\le\alpha/16\le \frac{\alpha}{2} \,p^*_\theta.
			\end{align}
		\end{subequations}

		Instead, if $p^\ast_\theta \le 1/8$, then by Chernoff's bound we can write the following:
		\begin{subequations} \label{eq:chernoff}
			\begin{align}
			\Pr\left(\tilde{\gamma}_\theta-p^\ast_\theta\ge \frac{\eps}{4} \right) &\leq  e^{-q (\epsilon/4)^2 \frac{1}{1-2p^*_\theta} \log(\frac{1-p^*_\theta}{p^*_\theta})}\le\\
			&
			\le e^{-2 \log(4/\alpha) \log(\frac{7}{8p^*_\theta})}= \\
			&
			= (\frac{8}{7} \,p^*_{\theta})^{2 \log(4/\alpha)} = \\
			&
			= \left( \frac{1}{e} \,{\frac{8}{7}\,e\, p^*_{\theta}}\right)^{2 \log(4/\alpha)}\leq \\
			&\le (e)^{ -2 \log(4/\alpha)} \frac{8}{7}\,e\,p^\ast_\theta \le\label{eq:chernoff_4}\\ 
			&\le \frac{\alpha}{16} \frac{8}{7} \,e \,p^\ast_\theta \le\\
			&\le \frac{\alpha}{4} \,p^\ast_\theta,
			\end{align}
		\end{subequations}
		%
		Moreover, we can write:
		\begin{subequations}
			\begin{align}
			\Pr\left(\tilde{\gamma}_\theta-p^\ast_\theta \le - \frac{\eps}{4} \right) &\leq  e^{-q (\epsilon/4)^2 \frac{1}{2(1-p^*_\theta)p_\theta^*}} =\\
			&
			=e^{-\frac{\log(4/\alpha)}{p^\ast_\theta}} =\\
			&=\left( e^{\frac{1}{p^\ast_\theta}} \right)^{\log\left(\frac{\alpha}{4}\right)}\leq\label{eq:lemma_3}\\
			&\le  \left(\frac{1}{p^\ast_\theta} e\right)^{\log\left(\frac{\alpha}{4}\right)}\leq \label{eq:lemma_4}\\ 
			&\le  \left(\frac{1}{p^\ast_\theta}\right)^{-1} e^{\log\left(\frac{\alpha}{4}\right)} = \label{eq:lemma_5}\\
			&= \frac{\alpha}{4} \, p^\ast_\theta,\label{eq:lemma_6}
			\end{align}
		\end{subequations}
		where in Equations~\eqref{eq:lemma_4} and \eqref{eq:lemma_5} we use that $e^x\ge e \,x$ and $\log(\alpha/4) < -1$ as $\alpha \in (0,1]$.
		Hence, we obtain the following inequality:		
		\[
		\sum_{i : |i/q-p^\ast_\theta|> \eps/4} \;\sum_{p \in \Q:p_\theta=i/q} \gamma_p =\Pr\left( |\tilde\gamma_\theta-p^\ast_\theta | > \frac{\eps}{4} \right)\leq \frac{\alpha}{2}p_\theta^\ast,
		\]
		which concludes the proof.
	\end{proof}

	The second lemma we introduce proves that, when  $\pvec_\theta$ is close to $\pvec^\ast$, then the utility of every receiver is close to the utility in $\pvec^\ast$ with high probability.
	\begin{lemma}\label{lemma: Th.3 lemma_2}
		Given $\pvec^\ast\in \P$, for each receiver $r \in R$, each state $\theta \in \Theta$ and each $i:|i/q-p^\ast_\theta|\le \eps/4$, it holds:
		\[
		\sum_{\pvec \in \mathcal{P}^r :p_\theta=i/q } \gamma_\pvec \ge \left(1-\frac{\alpha}{2}\right) \sum_{\pvec \in \Q : p_\theta=i/q } \gamma_\pvec,
		\]
		where $\gammavec$ is the distribution of $q$ i.i.d samples from $\pvec^\ast$.
	\end{lemma}
	\begin{proof}
		Fix $\bar \theta\in\Theta$, $ r \in R$ and $i$ with $|i/q-p^\ast_{\bar \theta}|\le \eps/4$.
		Then, let $\tilde t = \sum_{\theta}\tilde\gamma_\theta^{(\bar\theta,i)} u_r(\theta)$ and $t=\sum_{\theta} p^\ast_\theta u_r(\theta)$, where the notation $\tilde\gamma_\theta^{(\bar\theta,i)}$ is employed to denote the value of $p_\theta$ given that the random variable $\tilde\gammavec^{(\bar\theta,i)} \in \Q$ assumes value $\pvec$.
		First, we show that $|\, \EX[\,\tilde t\,]-t \,|\leq\eps/2$.
		This is equivalent to prove the following:
		\[|\sum_{\theta}u_r(\theta)\left(\EX[\tilde\gamma_\theta^{(\bar\theta,i)}]-p^\ast_\theta\right)|\leq \sum_{\theta}|\EX[\tilde\gamma_\theta^{(\bar\theta,i)}]-p^\ast_\theta|\le  \eps/2.\]
		Assume $i/q\geq p^\ast_{\bar\theta}$, then,
		\begin{equation}\label{eq:risultatino}
		\begin{split}
		&\sum_{\theta}|\EX[\,\tilde\gamma_\theta^{(\bar\theta,i)}\,]-p^\ast_\theta|=\\
		&=\frac{i}{q}-p^\ast_{\bar\theta}+ \sum_{\theta\neq\bar\theta}\left(p^\ast_\theta - \frac{p^\ast_\theta}{\sum_{\theta'\neq\bar\theta}p^\ast_{\theta'}}\,\left(1-\frac{i}{q}\right) \right)\leq\\
		&\leq  \frac{\eps}{4}+1-p^\ast_{\bar\theta}-1+\frac{i}{q}\leq \frac{\eps}{2}. \nonumber
		\end{split}	
		\end{equation}
		Analogously, if $i/q\leq p^\ast_{\bar\theta}$, we get that $\sum_{\theta}|\,\EX[\,\tilde\gamma_\theta^{(\bar\theta,i)}\,]-p^\ast_\theta\,|\leq \frac{\eps}{2}$.
		Now, we can exploit the fact that $|\,\EX[\,\tilde t\,]-t\,|\leq\eps/2$ to show that: $\Pr(|\,t-\tilde t\,|\geq \eps)\leq \Pr(\,|\,\tilde t-\EX[\,\tilde t\,]\,|\geq \eps/2)$ by the  triangular inequality. Then, we use the Hoeffding's inequality to bound the last term:
		\begin{equation*}\label{eq:lemma_hoeffding}
		\Pr(\,|\,\tilde t- \EX[\,\tilde t\,]\,|\geq \eps/2) \leq 2e^{-\frac{2q}{4}(\frac{\eps}{2})^2} \le
		2e^{-\log(4/\alpha)} 
		= \frac{\alpha}{2}
		\end{equation*}
		By definition of $\mathcal{P}^r$, this implies that $\Pr(\tilde \gammavec^{(\bar \theta, i)} \in \mathcal{P}^r) \ge 1 - \alpha/2$.
		Finally,
		\begin{align*}
		\sum_{\pvec \in \mathcal{P}^r : p_{\bar{\theta}}=i/q} \gamma_\pvec =& \Pr\left( \tilde{\gamma}_{\bar\theta}=\frac{i}{q}\right) \Pr \left( \tilde \gammavec \in \mathcal{P}^r \mid \tilde\gamma_{\bar\theta}=\frac{i}{q}\right)=\\=& 
		\Pr\left( \tilde{\gamma}_{\bar\theta}=\frac{i}{q}\right)\Pr\left(\tilde\gammavec^{(\bar\theta,i)}\in\mathcal{P}^r\right)\geq\\\geq&
		\left(1-\frac{\alpha}{2}\right) \Pr\left( \tilde{\gamma}_{\bar\theta}=\frac{i}{q}\right)=\\
		=&\left(1-\frac{\alpha}{2}\right) \sum_{\pvec \in \Q : p_{\bar{\theta}}=i/q} \gamma_\pvec.
		\end{align*}
	\end{proof}

	Before introducing the last lemma, we need some further notation. 
	More precisely, given a posterior, we partition the receivers in three sets, depending on their possible best-responses.
	We define the partition on the set of receivers induced by $\pvec^\ast \in \P$ as follows:
	\begin{itemize}
		\item $A = \left\{ r \in R : \sum_\theta p^\ast_\theta \, u_r(\theta)>0\right\}$,
		\item $B = \left\{ r \in R : \sum_\theta p^\ast_\theta \, u_r(\theta)<0\right\}$,
		\item $E = \left\{ r  \in R : \sum_\theta p^\ast_\theta \, u_r(\theta)=0\right\}$.
	\end{itemize}
	Similarly, any $q$-uniform posterior $\pvec \in \Q$ induces the following partition to the set of receivers when $\eps$-persuasiveness is adopted:
	\begin{itemize}
		\item $A_\eps = \left\{ r \in R : \sum_\theta p_\theta \, u_r(\theta)>\eps\right\}$,
		\item $B_\eps = \left\{ r \in R : \sum_\theta p_\theta \, u_r(\theta)<-\eps\right\}$,
		\item $E_\eps = \left\{ r  \in R : \sum_\theta p_\theta \, u_r(\theta)\in[-\eps,\eps]\right\}$.
	\end{itemize}
	Then, we define an auxiliary variable $\yvec^{\pvec} \in \C$ as follows:
	\begin{itemize}
		\item For every $r \in A$,  $y^{\pvec}_r =
		\left\{
		\begin{array}{ll} 
		c_0& \mbox{if } r \in  A_\eps \cup E_\eps \\
		c_1 & \mbox{otherwise}
		\end{array}
		\right..
		$
		\item For every $r \in B$,  $y^{\pvec}_r =
		\left\{
		\begin{array}{ll}
		c_1& \mbox{if } r \in  B_\eps \cup E_\eps \\
		c_0 & \mbox{otherwise}
		\end{array}
		\right..
		$
		\item For every $r \in E$,  $y^{\pvec}_r =
		\left\{
		\begin{array}{ll}
		b^{\pvec^\ast}_r& \mbox{if } r \in  E_\eps \\
		c_0& \mbox{if } r \in  A_\eps  \\
		c_1& \mbox{if } r \in  B_\eps \\
		\end{array}
		\right..
		$
	\end{itemize}
	Note that, by construction, $\yvec^\pvec$ is a valid action profile under $\eps$-persuasiveness. 
	Moreover, by the optimality of the $\eps$-persuasive best-response, the following holds for every posterior $\pvec$:
	\begin{equation}\label{eq:y^p .r.t. best response}
	\sum_\theta  p_\theta \,g_\theta(\bvec^{\pvec,\eps}) \ge \sum_\theta p_\theta \,g_\theta(\yvec^\pvec).
	\end{equation}
	Finally, let $\tilde \yvec^{(\theta,i)} \in \C$ be the random variable such that:
	\[\Pr(\,\tilde \yvec^{(\theta,i)}=\cvec\,)=\frac{\sum_{\pvec\in \Q, \pvec_\theta=i/q,\yvec^{\pvec}=\cvec} \gamma_\pvec}{\sum_{\pvec'\in \Q : \pvec'_\theta=i/q} \gamma_{\pvec'}}.\]

	Now, we introduce the last lemma we use to complete the proof. 
	This lemma proves that $\tilde \yvec^{\theta,i}$ are $\frac{\alpha}{2}$-noisy probability distributions around $\bvec^{\pvec^\ast}$.
	\begin{lemma}\label{lemma3}
		Given $\pvec^\ast \in \P$, for each $\theta \in \Theta$ and $i:|i/q-p^\ast_\theta| \le \eps/4$,  $\tilde \yvec^{(\theta,i)} \in \C$ is a $\frac{\alpha}{2}$-noisy probability distribution around $\bvec^{\pvec^\ast}$.
	\end{lemma}
	\begin{proof}
		We need to prove that for every receiver $r$, it holds $\Pr(\tilde y_r^{(\theta,i)} = b^{\pvec^\ast}_r) \ge 1-\alpha/2$.
It holds:
		\begin{equation}
		\begin{aligned}
		\Pr(\tilde y_r^{(\theta,i)} =& b^{\pvec^*}_r) = \frac{\sum_{\pvec\in \Q: p_\theta=i/q,y^{\pvec}_r=b^{\pvec^*}_r} \gamma_\pvec}{\sum_{\pvec'\in \Q : p'_\theta=i/q} \gamma_{\pvec'}}\ge \\
		&\ge \sum_{\pvec \in \mathcal{P}^r : p_\theta=i/q} \frac{\gamma_\pvec}{\sum_{\pvec'\in \Q:p'_\theta=i/q} \gamma_{\pvec'}} \ge \\
		&\ge \left(1-\frac{\alpha}{2}\right) \sum_{\pvec \in \Q : p_\theta=i/q } \frac{\gamma_\pvec}{\sum_{\pvec'\in \Q:p'_\theta=i/q} \gamma_{\pvec'}} =\\
		&= \left(1-\frac{\alpha}{2}\right).
		\nonumber
		\end{aligned}
		\end{equation}
This concludes the proof.
	\end{proof}

	Now, we are ready to prove Equation~(\ref{eq:general_bound}).
	\begin{subequations}
		\begin{align}
		&\sum_\theta \sum_{\pvec \in \Q} \gamma_\pvec \,p_\theta \,g_\theta(\bvec^{\pvec,\eps}) \ge \label{eq: th general bound; sender utility_1}\\
		& \hspace{2.4cm}\textnormal{(By restricting the set of posteriors.)}\nonumber\\
		&\ge \sum_\theta \sum_{i:|i/q-p^\ast_\theta|\le \eps/4} i/q \sum_{\pvec:p_\theta=i/q} \gamma_\pvec \, g_\theta(\bvec^{\pvec,\eps}) = \\
		&= \sum_\theta \sum_{i:|i/q-p^\ast_\theta|\le \eps/4} i/q \left(\sum_{\pvec:p_\theta=i/q} \gamma_\pvec\right)\\ &\hspace{1.9cm}\sum_{\pvec:p_\theta=i/q} \frac{\gamma_\pvec  }{\sum_{\pvec':p'_\theta=i/q} \gamma_{\pvec'}} g_\theta(\bvec^{\pvec,\eps})\ge \nonumber\\
		&\hspace{4.6cm}\textnormal{(By Inequality~\eqref{eq:y^p .r.t. best response}.)}\nonumber\\
		&\ge \sum_\theta \sum_{i:|i/q-p^\ast_\theta|\le \eps/4} i/q \left(\sum_{\pvec:p_\theta=i/q} \gamma_\pvec\right) \label{eq:optimality}\\ &\hspace{2.1cm}\sum_{\pvec:p_\theta=i/q} \frac{\gamma_\pvec  }{\sum_{\pvec':p'_\theta=i/q} \gamma_{\pvec'}} g_\theta( \yvec^\pvec)\ge  \nonumber\\
		&\hspace{0.4cm}\textnormal{(By stability of $g$ compared to $h$ and Lemma~\ref{lemma3}.)}\nonumber \displaybreak \\
		&\ge \sum_\theta \sum_{i:|i/q-p^\ast_\theta|\le \eps/4} i/q \left(\sum_{\pvec:p_\theta=i/q} \gamma_\pvec\right) \\
		&\hspace{4.1cm}\left(1-\frac{\alpha}{2}\beta\right) \,h_\theta(\bvec^{\pvec^\ast}) = \nonumber\\
		&=\left(1-\frac{\alpha}{2}\beta\right) \sum_\theta h_\theta(\bvec^{\pvec^\ast}) \\ & \hspace{2.6cm}\sum_{i:|i/q-p^\ast_\theta|\le \eps/4} i/q \sum_{\pvec:p_\theta=i/q} \gamma_\pvec \ge \nonumber\\
		&\ge \left(1-\frac{\alpha}{2}\beta\right) \sum_\theta h_\theta(\bvec^{\pvec^\ast})\\
		& \hspace{2.1cm}  \left(p^\ast_\theta -  \sum_{i:|i/q-p^\ast_\theta|\ge \eps/4} \sum_{\pvec:p_\theta=i/q} \gamma_\pvec\right) \ge \nonumber\\
		&\hspace{5cm}\textnormal{(By Lemma \ref{lemma: Th.3 lemma_1}.)}\nonumber\\
		&\ge \left(1-\frac{\alpha}{2}\beta\right) \sum_\theta h_\theta(\bvec^{\pvec^\ast})\, \left(1-\frac{\alpha}{2}\right)\,p^\ast_\theta = \\
		&= \left(1-\frac{\alpha}{2}\beta\right)\,\left(1-\frac{\alpha}{2}\right) \sum_\theta p^\ast_\theta \, h_\theta(\bvec^{\pvec^\ast}) \ge \\
		& \hspace{3.7cm}\textnormal{(By $\alpha = \eta \, \min\{1, 1/\beta\}$.)}\nonumber\\
		&\ge (1-\eta)\sum_\theta p^\ast_\theta \, h_\theta(\bvec^{\pvec^\ast}). \label{eq: th general bound; sender utility_Last}
		\end{align}
	\end{subequations}
	This concludes the proof.
\end{proof}

\theoremTwo*

\begin{proof}
	For every constant $\beta, \eps > 0,\, \eta \in (0,1]$, by Theorem~\ref{th: general bound}, we know that any posterior $\pvec^* \in \P$ guaranteeing a value $h(\pvec^*)$ can be expressed as a convex combination of $q$-uniform posteriors such that $\sum_{\pvec \in \Q} \gamma_\pvec \, g_\eps(\pvec) \ge (1 - \eta) \, h(\pvec^*)$.
	Therefore, given the optimal persuasive public signaling scheme $\phi^*$ optimizing $h$, we can decompose each posterior probability distribution $\pvec \in supp(\phi^*)$ into a convex combination of $q$-uniform posteriors and obtain an $\eps$-persuasive public signaling scheme $\phi_\eps$ maximizing $g$ that satisfies the inequalities stated in the theorem.
	Let $q=32\log\left(\frac{4}{\eta\,\min\{1;\,1/\beta\}}\right)/\eps^2$. Since, for a fixed number of samples $q$, the number of $q$-uniform probability distributions is at most $|\Theta|^q$, we can search for the $\eps$-persuasive public signaling scheme maximizing $g$ over probability distributions $\pvec \in \Q$, by solving the following Linear Program composed of $\mathcal{O}(|\Q|)$ variables and constraints:
	\begin{subequations}\label{eq:lp1}
		\begin{align*}
		\max_{\gammavec \in \Delta_\Q} & \sum_{\pvec \in \Q}  \gamma_\pvec \sum_{\theta \in \Theta} p_\theta \, g_\theta(\bvec^{\pvec,\eps})\\
		\textnormal{s.t.} &\sum_{\pvec \in \Q} \gamma_\pvec \, p_\theta =\mu_\theta \qquad\forall \theta \in \Theta
		\end{align*}
	\end{subequations}
	Finally, given the probability distribution on the $q$-uniform posteriors $\gammavec \in \Delta_{\Q}$, it is easy to derive the corresponding public signaling scheme $\phi_\eps$ by setting the following for every $\theta\in\Theta$ and $\cvec\in \C$:
	\[
	\phi_\eps(\theta,\cvec)=\sum_{\pvec\in\Q:\bvec^{\pvec,\eps} =\cvec} \gamma_\pvec  \,p_\theta.
	\]
\end{proof}

\section{Omitted Proofs on ``Comparative Stability and Voting Functions''}

\lemmaTwo*

\begin{proof}
	To prove the first part of the lemma, we need to show that for every voting profile $\bar \cvec \in \C$ and $\alpha$-noisy probability distribution $\yvec$ around $\bar \cvec$ with $\alpha \in (0,1]$,  the following inequality holds:
	\begin{equation}
	\label{inequality: stability}
	\EX_{\tilde \yvec \sim{} \yvec}\left[ W_\delta(\boldsymbol{\tilde{\yvec}})\right] = \sum_{\textbf{c} \in \C} y_{\textbf{c}} \, W_\delta(\textbf{c}) \geq W(\bar \cvec) \, \left(1-\frac{\alpha}{\delta}\right).
	\end{equation}
	Given that $W$ and $W_\delta$ assume values exclusively in $\{0,1\}$, Inequality \eqref{inequality: stability} is satisfied, independently from the chosen distribution $\yvec$, for all the voting profiles $\bar \cvec$ such that $W(\bar \cvec) = 0$. 
	Therefore, we can restrict our attention to the set of voting profiles such that $W(\bar \cvec) = 1$.
	Let $V_{c_0}(\cvec)=\{r \in R: c_r=c_0\}$ and $C^-=\{\cvec: |V_{c_0}(\cvec)|\le \ceil{(1-\delta) |R|/2} - 1 \}$.
	Then, for every $\yvec$, the following holds
	\begin{subequations}
		\begin{align*}
		\alpha |V_{c_0}(\bar \cvec)|&\ge \\
		&\ge \sum_{r \in V_{c_0}( \bar \cvec)} \sum_{\cvec \in \C:  c_r=c_1}  y_\cvec  \ge \\
		&\ge \sum_{\cvec \in \C^- } \sum_{r \in V_{c_0}( \bar \cvec):c_r=c_1} y_\cvec \ge\\
		&\ge \left[|V_{c_0}(\bar \cvec)| - \ceil{(1-\delta) |R|/2} - 1\right] \sum_{\cvec \in \C^-} y_\cvec  \ge \\
		&\ge  [|V_{c_0}(\bar \cvec)| - (1-\delta) |R|/2 ] \sum_{\cvec \in \C^-} y_\cvec  =\\
		&= [|V_{c_0}(\bar \cvec)| - \ceil{(1-\delta) |R|/2} ] (1-\EX_{\tilde \yvec \sim{} \yvec} \left[ W_\delta(\boldsymbol{\tilde{y}})\right]).
		\end{align*}
	\end{subequations}
	This implies that 
	\begin{subequations}
		\begin{align*}
		\EX_{\tilde \yvec \sim{} \yvec} \left[ W_\delta(\boldsymbol{\tilde{y}})\right] &\ge 1-\frac{\alpha |V_{c_0}(\bar \cvec)|}{|V_{c_0}(\bar \cvec)| - (1-\delta) |R|/2}=\\
		&=1-\frac{\alpha }{1 - (1-\delta) \frac{|R|/2}{|V_{c_0}(\bar \cvec)|}}\ge\\
		&\ge (1-\frac{\alpha}{\delta}) W(\bar \cvec),
		\end{align*}
	\end{subequations}
	where the last inequality follows from
	\[
		\frac{|R|/2}{|V_{c_0}(\bar \cvec)|}\le \frac{|R|/2}{\ceil{|R|/2}}\le 1 
	\] 
	and from $W(\bar \cvec)=1$ by assumption.

	Finally, to prove the second part of the lemma, we can employ Algorithm~\ref{Algorithm1} to show that for all $\bar{\textbf{c}} \in \C$ and for all probability distributions $\yvec$ around $\bar \cvec$ such that $\Pr_{\tilde \yvec \sim \yvec}[\tilde y_r =c_1 \land \bar{c}_r=c_0] \leq \alpha$, there is an $\alpha$-noisy probability distribution $\yvec'$ guaranteeing 
	\[ 
	\EX_{\tilde \yvec \sim \yvec}\left[ W_\delta(\boldsymbol{\tilde{y}})\right] \ge \EX_{\tilde \yvec \sim \yvec'}\left[ W_\delta(\boldsymbol{\tilde{y}})\right] \geq W(\bar{\textbf{c}})\,\left(1-\frac{\alpha}{\delta}\right).
	\]
	It is easy to see that $\yvec'$ is $\alpha$-noise: $\yvec'$ has null probability on all the voting profiles $\textbf{c}$ with a $r \in R$ such that $ c_r=c_0 \land \bar{c}_r=c_1$, \emph{i.e.}, $V_{c_0}(\cvec) \subsetneq V_{c_0}( \bar \cvec)$, while, $\Pr_{\tilde \yvec \sim \yvec'}[\tilde y_r =c_1 \land \bar{c}_r=c_0]= \Pr_{\tilde \yvec \sim \yvec}[\tilde y_r =c_1 \land \bar{c}_r=c_0] \le \alpha$.
	Moreover, since Algorithm~\ref{Algorithm1} moves probability mass from an action profile $\cvec$ to an action profile $\cvec'$ with $V_{c_0}(\cvec')\subseteq V_{c_0}(\cvec)$, it does not increase the expected value of $W_\delta$.
	This concludes the proof.
	\begin{algorithm}
		\caption{}
		\label{Algorithm1}
		For any $\textbf{c}$ s.t $V_{c_0}(\cvec) \subsetneq V_{c_0}(\bar \cvec)$ :
		
		\hspace{1cm} Take $\textbf{c}': V_{c_0}(\textbf{c}')  = V_{c_0}(\textbf{c}) \cap V_{c_0}(\bar{\textbf{c}})$ 
		
		\hspace{1cm}  $y'_{\textbf{c}'} \leftarrow{} y_{\textbf{c}'} + y_{\textbf{c}}$
		
		\hspace{1cm}  $y'_{\textbf{c}} \leftarrow{} 0$
		
		\vspace{0.2cm}
	\end{algorithm}
\end{proof}

\lemmaThree*
\begin{proof}

	We need to prove that the following inequality holds for all $\textbf{c} \in C^{|R|}$ and $\alpha$-noisy distribution $\yvec$ around $\textbf{c}$ with $\alpha \in (0,1]$.
	\[
		\EX_{\tilde \yvec \sim \yvec} \left[ \mathcal{W}_{\delta\delta}(\boldsymbol{\tilde{\yvec}})\right] \geq \mathcal{W}(\cvec) \left(1 - \frac{\alpha}{\delta^2} \right).
	\]
	The value of function $\W_{\delta\delta}$ depends on the values of all the district functions $W^d_{\delta}$.
	Indeed, given a voting profile $\textbf{c} \in \C$, the function $\W_{\delta\delta}$ assumes value $\W_{\delta\delta}(\textbf{c}) = \bar{W}_\delta( \,W^1_{\delta}(\, \textbf{c}^1),\ldots, W^D_{\delta}(\textbf{c}^D)\,)$. 
	Therefore, when it is perturbed by an $\alpha$-noisy probability distribution $\yvec$, its expected value can be expressed as:
	\[	
		\EX_{\tilde \yvec \sim \yvec} \left[ \mathcal{W}_{\delta\delta}(\boldsymbol{\tilde{y}})\right] =	\EX_{\tilde \yvec \sim \yvec}\left[\,\bar{W}_\delta( \,W^1_{\delta}(\boldsymbol{\tilde{y}^1}),\ldots, W^D_{\delta}(\boldsymbol{\tilde{y}}^D)\, )\,\right] .
	\]
	Lemma \ref{lemma: stability of W_delta wrt W} can be applied to all the couples of functions $W^d,W^d_{\delta}$, deriving the following inequality
	for every $d \in D$, $\textbf{c} \in C^{|R|}$, $\alpha \in (0,1]$: 
		\begin{equation}\label{eq:property}
		\Pr_{\tilde \yvec \sim \yvec}\left(W^d_{\delta} (\boldsymbol{\tilde{y}}^d)=c_1 \land W^d (\textbf{c}^d)=c_0\right)\le \alpha/\delta.
		\nonumber
		\end{equation} 
		If $W^d(\cvec^d)=c_1$, the above inequality is trivially satisfied, whereas, if $W^d(\cvec^d)=c_0$, we can write
		\begin{subequations}
			\begin{align*}
			&\Pr_{\tilde \yvec \sim \yvec}\left(W^d_{\delta} (\tilde{\yvec}^d)=c_1 \land W^d (\textbf{c}^d)=c_0\right)=\\
			&=\Pr_{\tilde \yvec \sim \yvec}\left(W^d_{\delta} (\boldsymbol{\tilde{y}}^d)=c_1\right) = 1-\EX_{\tilde \yvec \sim \yvec}\left[W_\delta(\tilde \yvec^d)\right] \le\\
			& \le 1-\left(1-\frac{\alpha}{\delta} \right)W(\cvec^d) =\alpha /\delta. 
			\end{align*}
		\end{subequations}
	We can use the above inequality and the fact that $\bar W$ is a majority-voting function to apply Lemma~\ref{lemma: stability of W_delta wrt W} to the couple of functions $\bar W$ and $\bar W_\delta$, thus showing the following:
	\begin{subequations}
		\begin{multline*}
		\EX_{\tilde \yvec \sim \yvec}\left[ \bar W_\delta \left(W^1_{\delta}( \tilde{y}^1),\dots,W^{|D|}_{\delta} (\tilde{\yvec}^{|D|})\right)\right] \geq\\
		\hspace{2cm}\ge\bar W\left(W^1( \cvec^1),\dots,W^{|D|} (\cvec^{|D|}) \right)  \left(1-\frac{\alpha}{\delta^2}\right).
		\end{multline*}
	\end{subequations}
	This implies that $\mathcal{W}_{\delta\delta}$ is $1/\delta^2$ stable compared to $\mathcal{W}$.
\end{proof}

\lemmaFour*

\begin{proof}
	The proof follows the same steps of the proof of Lemma~\ref{th: general bound}.
	In the following, we just highlight the differences between the two proofs.
	In the steps from Equation~\eqref{eq: th general bound; sender utility_1} to Equation~\eqref{eq: th general bound; sender utility_Last}, we remove the summation over the states of nature.
	All the other steps hold, except for Equation~\eqref{eq:optimality}. 
	Indeed, since $\eps$-best response is computed maximizing the expected utility of the sender, there are no guarantees that for each state of nature $\theta$ it holds $g_\theta(\bvec^{\pvec,\eps}) \ge g_\theta(\yvec^{\pvec})$.
	However, since $W_\delta$ is state-independent and monotone non-decreasing in the number of receivers that vote for $c_0$, the best response $\bvec^{\pvec,\eps}$ is given by $b_r^{\pvec,\eps} = c_0$ for all the voters with utility $u_r(\theta)\ge -\eps$. 
	Thus, we are guaranteed that, for every $\yvec^\pvec \in \C$, it holds $W_\delta(\yvec^\pvec) \le W_\delta(\bvec^{\pvec,\eps})$ independently from the state of nature $\theta$.
	Taking into account Lemma \ref{lemma: stability of W_delta wrt W}, the derivation is straightforward.
\end{proof}

\section{Omitted Proofs on ``Computing a Semi-Public Signaling Scheme''}

\theoremThree*

\begin{proof} 

	Let $q=32\log\left(\frac{4}{\,\eta\,\delta\,}\right)/\eps^2$ and $\Q \subset \Delta_\Theta$ be the set of $q$-uniform probability distributions on $\Theta$. 
	We show that, given the optimal semi-public signaling scheme $\phi^*$, there is a solution $\phi_\epsilon$ to LP~\eqref{lp:semi_efficient} with $\W_\delta(\phi_\eps) \ge (1 - \eta)  \W(\phi^*)$.
	Given the signaling scheme $\phi^*$, let:
	\begin{itemize}
	%
	%
	\item $\betadue^*_{d,\theta}$ be the probability that $c_0$ wins in district $d$ when the state of nature is $\theta$ and 
	\item $\alpha^*_\theta$ be the probability that $c_0$ wins in at least $K_d$  when the state of nature is $\theta$.
	\end{itemize}
	Then, as showed in Theorem~\ref{th: optimal_private}, the probability such that $c_0$ wins in at least $K_D$ districts with state of nature $\theta$ is:   
	\begin{equation}\label{eq:alpha}
	\alpha^*_\theta= \min\left\{\min_{m\in\{0,\ldots,K_D-1\}} \frac{1}{K_D-m}\qvar_{\theta,m}; \, 1 \, \right\},
	\end{equation}
	where $\qvar_{\theta,m}$ is the sum of the lowest $|R^d|-m$ elements in the set $\{\betadue^*_{d,\theta}\}_{d\in D}$.
	We show that there is a solution to LP~\eqref{lp:semi_efficient} with $\betadue^\delta_{d,\theta}\ge (1-\eta) \betadue^*_{d,\theta}$ for every $d$ and $\theta$. 
	Since the value of each  $\betadue_{d,\theta}$ is reduced by a multiplicative factor $(1-\eta)$,  Equation~\eqref{eq:alpha} implies that $\alpha_\theta \ge (1-\eta) \alpha^*_\theta$ and $\sum_{\theta}\mu_\theta\alpha_\theta\ge (1-\eta) \sum_{\theta}\mu_\theta\alpha^*_\theta$.\footnote{See Theorem~\ref{th: optimal_private} for details on how LP~\eqref{lp:semi_efficient} computes $\alpha_\theta$ from $\beta^\delta$.}

	Hence, we conclude the proof showing that  $\betadue^\delta_{d,\theta} \ge (1-\eta)\betadue^*_{d,\theta}$ for every $d$ and $\theta$.
	Let:
	\begin{itemize} 
	\item $\phi^*_d$ be the marginal probabilities of the signaling scheme $\phi$ restricted to the receivers in district $d$,
	\item $\gammavec^* \in \Delta_\P$ be the probability distribution on posteriors induced by $\phi^*_d$,
	\item $\gammavec^\pvec \in \Delta_\P$ be the probability distribution on $q$-uniform posteriors obtained decomposing a posterior $\pvec$ as prescribed by Lemma~\ref{corollary: th_general bound}, and  
	\item $\gammavec^d \in \Delta_\Q$ be the distribution on $q$-uniform posteriors obtained by decomposing each posterior induced by $\phi^*_d$ as in Lemma~\ref{corollary: th_general bound}, \emph{i.e.,} $\gamma^d_\pvec=\sum_{\pvec' \in supp(\phi^*)} \gamma^*_{\pvec'} \, \gamma^{\pvec'}_\pvec$ for every $\pvec$.
	\end{itemize}
	We conclude proving that $\gammavec^d$ is a $q$-uniform distribution that induces a $\betadue^\delta_{d, \theta}\ge (1-\eta)\betadue^*_{d,\theta}$ for every $\theta$.
	\begin{equation}
	\label{eq: th_semi_efficient_proof}
	\begin{aligned}
	&(1-\eta)\betadue^*_{d,\theta}=\\
	&=(1-\eta)\sum_{\substack{\pvec \in supp(\phi^*_d)}} \frac{\gamma_\pvec^* \, p_\theta}{\mu_\theta} \mathbb{I}\left(W^d(\bvec^\pvec) = c_0\right) \le\\
	& \hspace{5cm}(\textnormal{by Lemma~\ref{corollary: th_general bound}})\\
	&\le \sum_{\pvec \in supp(\phi^*_d)} \frac{\gamma_\pvec^*}{\mu_\theta}  \sum_{\pvec' \in \Q} \gamma^\pvec_{\pvec'} p'_\theta \mathbb{I}\left(W_\delta(\bvec^{\pvec',\eps})=c_0\right) = \\
	&=\sum_{\pvec' \in \Q} \frac{p'_\theta}{\mu_\theta} \mathbb{I}\left(W_\delta(\bvec^{\pvec',\eps})=c_0\right)\sum_{\pvec \in supp(\phi^*_d)} \gamma_\pvec^* \gamma^\pvec_{\pvec'}   =\\
	&= \sum_{\pvec \in \Q} \frac{\gamma_\pvec^d  p_\theta}{\mu_\theta}\mathbb{I}\left(W_\delta(\bvec^{\pvec,\eps})=c_0\right) =\\
	&=\betadue^\delta_{d,\theta}. \nonumber
	\end{aligned}
	\end{equation}	
	This concludes the proof.
\end{proof}